\documentclass[letterpaper, 11pt]{article}

\usepackage{graphicx}
\usepackage{latexsym}
\usepackage{amsmath} 	
\usepackage{amssymb}
\usepackage{rotating}
\usepackage{amsthm}
\usepackage{fullpage}

\newtheorem{theorem}{Theorem}
\newtheorem{lemma}{Lemma}
\newtheorem{proposition}{Proposition}
\newtheorem{corollary}{Corollary}
\newtheorem{definition}{Definition}

\newcommand{\Tr}{\mbox{Tr}}

\title{\bf A Limit on Non-Locality Distillation}

\author{Dejan D. Dukaric$^{\star}$\qquad Stefan Wolf$^{\dag}$ \\[1cm] \normalsize\sl $^{\star}$ Department of Computer Science \& Department of Physics \\ \normalsize\sl ETH Zurich, \textsc{Switzerland}\\[-0.1cm]  \normalsize \texttt{ddukaric@ethz.ch}  \\[0.5cm] \normalsize\sl $^{\dag}$ Department of Computer Science\\ \normalsize\sl ETH Zurich, \textsc{Switzerland}\\[-0.1cm] \normalsize \texttt{wolf@inf.ethz.ch}\\[0.5cm]}

\begin{document}

\maketitle


\begin{abstract}
\noindent
Non-local correlations are not only a fascinating feature of quantum theory, but 
 an interesting resource for information processing, for instance in
communication-complexity theory or  cryptography. An important question in this context
is whether the resource can be distilled: Given a large 
amount of {\em weak\/} non-local correlations, is there a method to obtain {\em strong\/} non-locality  using  local operations and shared randomness? 
We partly answer this question by {\em no\/}: symmetric (or isotropic) 
CHSH-type non-locality which is not ``super-strong,'' but achievable by  measurements on certain quantum states, has at best very limited distillability by any non-interactive classical method. This strongly extends and generalizes what was previously known, namely 
that there are two limits that cannot be overstepped: The Bell and Tsirelson bounds. 
Our results imply that there must be an infinite number of such bounds.
  
A noticeable feature of our proof of this --- purely classical --- statement is that it is ``quantum mechanical'' in the 
sense that (both novel and known) facts from quantum theory are used in a crucial way to obtain the claimed 
results. One of these results --- of independent interest --- is that certain mixed entangled states 
cannot be distilled at all without communication. Weaker statements, namely {\em limited\/} distillability,
have been known for {\em Werner states}.  
\end{abstract}

\newpage

\section{Motivation and Main Result}

\subsection{Non-Locality}

\noindent
{\it What is non-locality?}

\noindent
Linearity of quantum physics allows for a phenomenon called {\em entanglement\/}: The joint state 
of two  possibly distant  parts of a system  cannot be described by  the states 
of the parts. Even when the two parts are measured in a way that permits no signal transmission from one measurement event
to the other | such two events are  called
{\em space-like separated\/} |, the outcomes can be correlated. Although such a correlation 
may be surprising from a {\em physical\/} point of view, it need not be in terms of {\em information\/}:
The distant particles could have shared classical information since they were generated together.
Actually, this is what motivated Einstein, Podolsky, and Rosen \cite{epr} in 1935
to postulate that quantum physics was incomplete and should be augmented by this information,
called {\em hidden variables}, determining the outcomes of all potential measurements.
The true surprise came about $30$  years later, when John Bell \cite{bellInequality}  showed that quantum correlation 
{\em cannot\/} be explained by shared classical information in principle. 

\ \\ 
\noindent
{\it How can this be? Measurement outcomes are classical random variables, are they not?}

\noindent
It can be if one looks at correlations of two-partite {\em systems\/}
with a given input/output behavior instead of just random 
variables. This makes sense physically, too: the input corresponds the the choice of the measurement basis.

A {\em bipartite system\/} consists of two ends, usually called Alice and Bob, at each of which an input ($x$ and $y$)
is given and an output obtained ($a$ and $b$).
There are different degrees of correlation between the two ends: 
A system is called {\em signaling\/} if it allows for message transmission between Alice and Bob, and {\em local\/} 
if its behavior can be explained by shared classical information. Now, {\em Bell's theorem}, which 
has been celebrated as one of the ``most profound discoveries in science'' (Stapp), implies 
that certain entangled quantum states behave, when measured, non-locally (but, at the same 
time, non-signaling). Non-local behavior is, hence, a  reality; and an experimentally 
verified one. 

\ \\ 
\noindent
{\it What is known about non-locality, besides that it exists?}

\noindent
Non-locality  is being intensively studied, not only as a fascinating {\em phenomenon}, but also   as a {\em resource\/} 
for information processing. 
Many central properties  remain unknown; an example being whether there 
is a method to ``distill'' or ``amplify'' non-local correlations.
In this note, we give a partial  answer to the question: 
A specific but central  
type of non-locality  has at most very limited distillability.

\ \\ 
\noindent
{\it What can non-locality be useful for?\/} 

\noindent
Two examples: First, it can {\em replace\/} 
communication in certain contexts | although it does not {\em allow\/} for communication. Second, 
it implies cryptographic confidentiality. 
The fact that the outputs of a {\em non\/}-local system are {\em not\/} completely pre-determined
as classical information  
 implies lack of knowledge by any possible 
adversary~\cite{kent}!
Note that the resulting cryptographic security follows directly from the observed 
joint correlations, is {\em device-independent}, and does not even rely on the fact 
that the adversary is limited by quantum physics; it only depends on the non-signaling 
postulate of relativity. 

\ \\ 
\noindent
{\it But then, why is the problem of non-locality distillation important and interesting?}

\noindent
Generally, the stronger the non-local correlations, the more useful they are. 
For instance, it has been shown~\cite{world} that the existence of non-locality which is super-quantum to some extent
would collapse communication complexity: Every Boolean function could be computed in a distributed
way with just a constant number of communicated bits. A second example is from cryptography: The confidentiality
that can be derived from non-locality is stronger if the non-locality is. 
It is, by these reasons,
a  natural question whether it can be distilled (see also \cite{nonLocalCorr} where non-locality distillation is mentioned to be an interesting problem which should be investigated).

\ \\ 
\noindent
{\it For which types of non-locality do the results apply? Are there other types which 
are distillable?}

\noindent
We investigate a special type of non-locality, namely {\em isotropic approximations 
to non-local boxes}. The latter are idealizations of systems violating the most prominent 
Bell inequality called CHSH (after Clauser, Horne, Shimony, and Holt) . Isotropic
boxes are symmetric and, therefore, an important special case. They are, actually, 
the worst case when it comes to distillability: If {\em they\/} can be distilled,
then {\em all\/} boxes of the same strength can. Our result states that the distillability 
of such systems is, at most, very limited (and suggests that it is, in fact, zero). This 
complements a result of~\cite{FosterBox}, stating that certain (non-isotropic) systems 
can be distilled. This may indicate that CHSH violation is a good measure for non-locality 
only in the isotropic case.

\subsection{Our Contribution}
The most important and prominent 
 type of non-locality is  {\em CHSH-type\/} non-locality~\cite{CHSH}. There are two main reasons why we investigate this type of non-locality. On the one hand, it is the only possible non-locality of systems
with binary inputs and outputs. On the other hand, for binary input/output systems the non-locality can be quantified by a single number (see equation (\ref{CHSHeqn})), whereas for systems with more inputs/outputs it is not clear how to quantify non-locality by a single number and therefore, to find a meaningful definition of non-locality distillation for such systems.
The maximal possible CHSH-type 
non-locality has been modeled as a {\em Popescu-Rohrlich (PR)-box\/}~\cite{NonlocalityAxiom} or {\em non-local (NL) box\/}. Such a PR-box has input bits $x, y \in \{0,1\}$ and output bits $a,b \in \{0,1\}$ on Alice's and Bob's side, respectively, and computes the relation
\[
x \wedge y = a \oplus b \ .
\]
In other words, Alice and Bob have to output different bits if and only if both of them obtain $1$ as their input. For all other input pairs they have to output the same bit.
Furthermore, the output bits $a$ and $b$ must be unbiased (see Section \ref{sectionSystems} for a definition), otherwise the bipartite system is signaling.
It has been shown that a system which has this behavior with probability $>75\%$, averaged over all input pairs, is non-local; on the other hand, entangled
quantum states can achieve roughly $85\%$. We call a system a $q$-approximation to the PR-box if it satisfies the above described behavior, averaged over all input pairs of Alice and Bob, with probability $q$. 
The question we address in this paper is: Given a number of 
$q$-approximations to the PR-box, is it possible to obtain, without communication, a $p$-approximation 
for some $p>q$? In other words, 
we investigate {\em Non-Locality Distillation Protocols\/} (NDPs for short), which are two-party protocols between Alice and Bob. An NDP takes as input $x, y \in \{0,1\}$ and outputs $a, b \in \{0,1\}$ such that a PR-box is approximated with probability $p$. Alice and Bob are allowed to perform arbitrary local operations, have shared randomness, can use $q$-approximations of PR-boxes (where $q < p$), but are not allowed to communicate (see Section \ref{NDP} for a detailed description of NDPs).
We will prove the following impossibility result on the existence of NDPs. 
 \\

\noindent
{\bf Main Result.}
{\em There exists no non-locality distillation protocol which uses isotropic $q$-approxi\-mations of PR-boxes, and has the same input/output behavior as a $p$-approximation of a PR-box such that $p > q + g(q)$ and
\begin{equation}
\label{qm}
 \frac{3}{4} \leq q \leq \frac{2+\sqrt{2}}{4} ~,
\end{equation}
where $g(q)$ is the distillability gap (see Theorem \ref{resultGeneralDistillation}).
}
\newline	

This statement is the first result about non-locality distillation besides Bell's and Tsirelson's 
limits, which no algorithm can cross (Section~\ref{trivialBounds} covers the known bounds).			
Why the restriction in the statement? The reason lies in a particularity of its proof: it is ``quantum-physical''; 
this is remarkable since the statement is purely classical.\footnote{By a ``quantum-physical proof'', 
we do not mean that whether the statement is correct or not depends on whether quantum theory 
is an accurate description of nature or not. We mean that we make crucial use of mathematical 
theorems that deal with quantum-physical notions and that make sense, and have been proven,
in the context of quantum theory. Therefore, the proof depends only on the soundness of quantum mechanics as a mathematical theory.}
An example of such a statement is that the {\em fully entangled fraction\/} | i.e., proximity to a maximally entangled state |  of 
certain bi-partite states cannot be increased non-interactively. Such a result has not been known before
for any (non-trivial) state. For a  different class of states, the {\em Werner states}, it has been shown that the 
possibility of non-interactive purification is {\em limited\/}; for our states, it is {\em zero}.
\ \\

\noindent
In a nutshell, our proof works as follows: 
\begin{itemize}
\item
The proof is by contradiction: 
We assume that an NDP exists which uses a finite number of isotropic $q$-approximations of PR-boxes | 
for some $q$ that is quantum-physically achievable, i.e., satisfying (\ref{qm}) | and has the same input/output behavior as
a $p$-approximation of a PR-box, for some $p > q + g(q)$. 

\item
This classical NDP can be simulated by local quantum operations, where we think of the $q$-boxes as coming from measuring a certain mixed quantum state $\Omega_{\alpha}$.

\item 
Measurements can be postponed to the end, and the measurement operators
 on Alice's and Bob's sides can be simultaneously block-diagonalized, in blocks of size at most $2 \times 2$. It is, therefore,
 sufficient to consider measurements on qubits only.

\item
Using the following two facts, we conclude that the CHSH-type non-locality resulting from the quantum NDP
is not much stronger than the original.

\begin{enumerate}
\item
It is impossible to obtain, from many copies of $\Omega_{\alpha}$, a state with a higher fully entangled fraction (i.e., proximity to a fully entangled state).

\item
By transforming this result from a statement about the impossibility of entanglement distillation to a statement about the impossibility of non-locality distillation we lose something (but not too much, namely the distillability gap $g(q)$) as there is no one-to-one correspondence between non-locality (CHSH violation) and entanglement (fully entangled fraction).
\end{enumerate}

\end{itemize}

Note that the statement we prove resembles a similar result on {\em correlations}, i.e., 
the impossibility of non-interactive correlation distillation~\cite{DrYang}: Without 
communication, no highly correlated bit can be obtained from an arbitrary number 
of weakly correlated bits.

\section{Definitions and Preliminaries}

\subsection{Systems and Correlations}
\label{sectionSystems}

In order to explain the essence of non-locality, we introduce the notion of {\em two-partite 
systems}, defined by their joint input-output behavior $P(a,b \vert x,y)$ with inputs $x \in \mathcal{X}$, $y \in \mathcal{Y}$ and outputs $a \in \mathcal{A}$, $b \in \mathcal{B}$ (see Figure \ref{boexli}). A party receives its output immediately after giving its input, independently of whether the other has given its input already. Depending on the concrete form of $P(a,b \vert x, y)$, the input/output behavior can be simulated by a physical process, namely by measurements on entangled quantum states, where $x$ and $y$ denote which measurements Alice and Bob select and $a$ and $b$ are the corresponding measurement results of Alice and Bob. 

\begin{figure}[h]
\begin{center}
\input{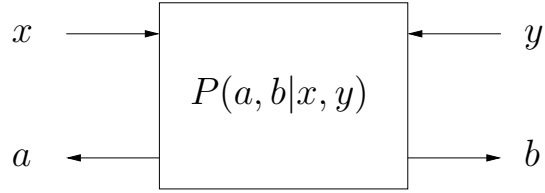}
\end{center}
\caption{A two-partite system.}
\label{boexli}
\end{figure}

If we have binary inputs and outputs, i.e., $\mathcal{X} = \mathcal{Y} = \mathcal{A} = \mathcal{B} = \{0,1\}$, the behavior of such as system can be represented by a four-valued function $P(a,b,x,y) \equiv P(a,b \vert x,y) : \{0,1\}^4 \rightarrow [0, 1]$, such that $\sum_{a,b \in \{0,1\}} P(a,b \vert x,y) = 1$ for all $x,y \in \{0, 1\}$. By the notation $P(a,b \vert x,y)$ we want to indicate that the outputs $a,b$ depend on the inputs $x,y$. One can view $P(a,b \vert x,y)$ as a conditional probability distribution or a stochastic matrix which assigns to each input/output pair a certain probability.

The behavior on Alice's side alone is described by $P_A(a,x) \equiv P_A(a \vert x): \{0,1\}^2 \rightarrow [0,1]$ such that $\sum_{a \in \{0,1\}} P_A(a \vert x) = 1$ for all $x \in \{0,1\}$, and similarly for Bob. If $P(a,b \vert x,y)$ is non-signaling the local behavior of Alice can be computed from the joint behavior by $P_A(a \vert x) = P_A(a \vert x,y) = \sum_{b} P(a,b \vert x,y)$.  We call Alice's output bit \emph{unbiased} if $P_A(a \vert x) = 1/2$, for $a, x \in \{0, 1\}$. We classify systems by the correlation they introduce | or, equivalently, by the resource that is required 
to explain the joint behavior of its parts.

\begin{definition}
\label{localDef}
{\rm 
Consider a two-partite system which is characterized by $P(a,b \vert x,y)$.  
\begin{itemize}
\item The system $P(a,b \vert x,y)$ is {\em independent} if there exist $P_A(a \vert x)$ and $P_B(b \vert y)$ such that
\[
P(a,b \vert x,y) = P_A(a \vert x) \cdot P_B(b \vert y) ~~ \forall a,b,x,y \in \{0,1\} \ .
\]

\noindent
\item The system $P(a,b \vert x,y)$ is {\em local\/} if there exist $\{ P^i_A(a \vert x) \}_{i=1}^n$ and $\{ P^i_B(b \vert y) \}_{i=1}^n$ such that
\[
P(a,b \vert x,y) = \sum_{i=1}^n{w_i \cdot P^i_A(a \vert x) \cdot  P^i_B(b \vert y)} ~~ \forall a,b,x,y \in \{0,1\}
\]
for some weights $w_i\geq 0$, $i = 1,...,n$. A system $P(a,b \vert x,y)$ is {\em non-local\/} if it is not local.

\noindent
\item The system $P(a,b \vert x,y)$ is {\em non-signaling\/} if 
\begin{eqnarray}
\sum_{b \in \{0,1\}} P(a,b \vert x,y) &=& \sum_{b \in \{0,1\}} P(a,b \vert x, y') ~~ \forall a,x,y,y' \in \{0,1\}~, \nonumber \\
\sum_{a \in \{0,1\}} P(a,b \vert x,y) &=& \sum_{a \in \{0,1\}} P(a,b \vert x', y) ~~ \forall b,x,x',y \in \{0,1\}~. \nonumber
\end{eqnarray}
Consequently, a system $P(a,b \vert x,y)$ is {\em signaling\/} if it is not non-signaling.

\end{itemize}
}
\end{definition}

In terms of {\em classical\/} resources, these categories correspond to {\em no resources\/} 
at all, {\em shared information\/} and {\em message transmission\/} (i.e., signaling), respectively. In this paper, we
focus on systems in-between, i.e., which are non-local but at the same time non-signaling.

\subsection{CHSH-Type Non-Locality}
\label{sectionCHSH}
The only non-locality a system can have if the inputs and outputs are binary, is so-called
{\em CHSH-type non-locality}. 
A system is CHSH-non-local exactly if it allows for simulating the above-described 
PR-box with a probability superior to $75\%$ \cite{bellInequality}. On the other hand, appropriate measurements on quantum states achieve $(2+\sqrt{2})/4\approx 85\%$, 
but no more than that~\cite{TsirelsonBound}. It is still an open question why quantum physics is non-local, but not maximally so.

In order to precisely quantify the CHSH-type non-locality of $P(a,b \vert x,y)$, we first define four correlation functions:
\begin{equation}
\label{corrFunc}
  E_{x,y} = P(0,0 \vert x,y) + P(1,1 \vert x,y) - P(0,1 \vert x,y) - P(1,0 \vert x,y) ~,~  x,y \in \{0,1\} ~.
\end{equation}
Without loss of generality, one can assume that $E_{0,0}$, $E_{0,1}$, $E_{1,0} \geq 0$, because there always exist local reversible transformations on $P(a,b \vert x,y)$ which realize these conditions. The CHSH-type non-locality of $P(a,b \vert x,y)$ can then be written as
\begin{equation}
\label{CHSHeqn}
NL[P(a,b \vert x,y)] = E_{0,0} + E_{0,1} + E_{1,0} - E_{1,1} ~.
\end{equation}
Algebraically speaking, the maximal  value equation (\ref{CHSHeqn}) can attain is $4$, corresponding to Alice and Bob computing the relation $x \wedge y = a\oplus b$, i.e., simulating  perfectly a PR-box. It is not difficult to see that the behavior of $P(a,b \vert x,y)$ with $NL[P(a,b \vert x,y)] \in [2, 2\sqrt{2}]$ corresponds to a $q$-approximation of a PR-box with $q = \frac{1}{2} + \frac{1}{8} \cdot NL[P(a,b \vert x,y)] \in [3/4, (2 + \sqrt{2})/4]$. Therefore, as a $3/4$-PR-box can be simulated by shared randomness only, the joint input/output behavior of Alice and Bob is local if $NL[P(a,b \vert x,y)] \leq 2$. For the rest of this paper we will use the shorter notation $q$-PR-box for a $q$-approximation of a PR-box. Furthermore, we will use the notation $NL[\rho_{AB}]$ to denote also the maximal attainable non-locality of a bipartite-quantum state $\rho_{AB}$ by appropriate measurements.

A $q$-PR-box can be simulated by {local} quantum processes and shared entanglement between Alice and Bob only if  $q \leq \frac{2 + \sqrt{2}}{4}$. This is, 
however, not a sufficient condition:  There exist $q$-PR-boxes with $q \leq \frac{2 + \sqrt{2}}{4}$ which cannot be simulated quantum-physically~\cite{NLnotQuantum}.

In this paper, we will concentrate on PR-boxes with a special error behavior, so called \emph{isotropic} PR-boxes. We call a $q$-PR-box isotropic if Alice and Bob compute $x \wedge y = a \oplus b$ with probability $q$ and $x \wedge y = \overline{a \oplus b}$ with probability $1-q$. Furthermore, the probability that Alice and Bob output $a=0,b=0$ or $a=1,b=1$ is equal for all input pairs, and similarly for the outputs $a=0,b=1$ or $a=1,b=0$.
More formally, we call a $q$-PR-box {\em isotropic\/} if it satisfies  the following constraints on its  correlation functions: $E_{0,0} = E_{0,1} = E_{1,0} = - E_{1,1} \geq 0$.
Additionally, we need $\sum_{b \in \{0,1\}} P(a,b \vert x,y) = P_A(a \vert x) = 1/2$ for all $a,x,y \in \{0,1\}$ and $\sum_{a \in \{0,1\}} P(a,b \vert x,y) = P_B(b \vert y) = 1/2$ for all $b,x,y \in \{0,1\}$, i.e., the outputs are unbiased.
An isotropic $q$-PR-box is quantum-physically realizable if and only if $q\leq (2+\sqrt{2})/4$ (see Lemma \ref{wernerStateSimulation}).

\subsection{Non-Locality-Distillation Protocols}
\label{NDP}
Roughly speaking, a non-locality-distillation protocol is a | classical | bipartite protocol consisting of 
$q$-PR-boxes plus local operations, and having as a whole the behavior of a $p$-PR-box for some $p>q$.

An NDP $\mathcal{N}^{p}_{q}$ takes as input the bits $x$ and $y$ on Alice's and Bob's side, respectively. The outputs of $\mathcal{N}^{p}_{q}$ are the bits $a$ and $b$, where $\mathcal{N}^{p}_{q}$ must have the same input/output behavior as a $p$-PR-box. During the protocol execution, Alice and Bob are allowed to use arbitrary local operations, shared randomness, and $q$-PR-boxes, but are not allowed to communicate. Because a box is non-signaling, it can be modeled as two local \emph{terminals} on Alice's and Bob's side, respectively. A terminal is locally described by its input-output behavior. By using this language of terminals, a protocol execution consists of Alice and Bob applying local operations and local terminals. Moreover, one can easily handle protocols where inputs to boxes are delayed or boxes are intertwined. By the latter, we mean the situation that on Alice's side, 
the input to some Box~2 depends on the output of a Box~1, whereas on Bob's side, 
the dependence is opposite, i.e., the input to Box~1 depends on Box~2. Note that causality is not violated since all boxes are non-signaling; boxes can answer  
outputs to inputs locally, even if no input has been given on the other end yet.

\section{Known Limits to  Non-Locality Distillation}
\label{trivialBounds}

Lemma~\ref{corollarryClassicalBound} states that it is not possible to create non-locality from locality.
\begin{lemma}{\bf (Bell \cite{bellInequality})}
\label{corollarryClassicalBound}
There exists no non-locality distillation protocol $\mathcal{N}^{p}_{q}$ with $q \leq 3/4$ and $p > 3/4$.
\end{lemma} 

\noindent
Lemma~\ref{corollarryQuantumBound}, on the other hand, gives a bound on what is quantum-physically achievable (the proof can be found in the appendix).
\begin{lemma}{\bf (Tsirelson \cite{TsirelsonBound})}
\label{corollarryQuantumBound}
There exists no non-locality distillation protocol $\mathcal{N}^{p}_{q}$, where the $q$-PR-boxes must be quantum-physically realizable, with $q \leq \frac{2 + \sqrt{2}}{4}$ and $p > \frac{2 + \sqrt{2}}{4}$. 
\end{lemma}
Note that the argument in the proof of Lemma \ref{corollarryQuantumBound} holds only for quantum-physically realizable approximations of PR-boxes. At first sight, one might think that \emph{every}
$q$-PR-box with $q \in [\frac{3}{4}, \frac{2 + \sqrt{2}}{4}]$ can be simulated quantum mechanically. But this is \emph{not} the case \cite{NLnotQuantum}. Actually, there exist $(3/4 + \varepsilon)$-PR-boxes, for any $0 < \varepsilon < 0.125$, which can be distilled up to $p = 0.875 > \frac{2 + \sqrt{2}}{4}$. Of course, these boxes cannot be simulated by quantum mechanics, as otherwise Tsirelson's bound would be crossed. Furthermore, there even exist non-isotropic PR-boxes which can be simulated by quantum mechanics \emph{and} can be distilled \cite{FosterBox}.

Finally, by an easy argument one can show that perfect boxes cannot be obtained from imperfect ones.
\begin{lemma}
Let $\varepsilon>0$.
There exists no non-locality distillation protocol $\mathcal{N}^{p}_{q}$ with $q \leq 1 - \varepsilon$ and $p = 1$.
\end{lemma}

\section{A New Result on Non-Locality Distillation}
\label{sectionImpossibilityDistillation}

In this section we will prove our main result,  Theorem~\ref{resultGeneralDistillation},   stating that it is impossible to distill more than a certain amount of non-locality (which will depend on $q$) from isotropic $q$-PR-boxes with $\frac{3}{4} \leq q \leq \frac{2 + \sqrt{2}}{4}$.
Similarly to the proof of Lemma \ref{corollarryQuantumBound}, we will use quantum mechanics to prove this \emph{classical} statement. As a preparation, 
we  establish some results about the correspondence between \emph{CHSH-type non-locality} and the \emph{fully entangled fraction}.

\subsection{CHSH Violation, Fully Entangled Fraction, and Weakly Entangled Mixed States}

The CHSH-type non-locality will be our measure of non-locality, whereas fidelity is a measure of entanglement,
more precisely, of proximity   between quantum states, defined by 
\begin{equation}
\label{fidelityDef}
F(\rho,\sigma) = {\rm Tr}^{2}\left(\sqrt{\sigma^{1/2} \rho \sigma^{1/2}}\right)\ ,
\end{equation}
where $\rho$ and $\sigma$ are mixed  states. Equation (\ref{fidelityDef}) simplifies to $F(\rho, \left|\psi\right\rangle \left\langle \psi \right|) = \left\langle \psi\right| \rho \left| \psi \right\rangle$,
if $\sigma = \left| \psi \right\rangle \left\langle \psi \right|$ is a pure state. In particular, we are interested in the quantity
\begin{equation}
\label{fidelityPure}
F(\rho) = \max_{U_{A},U_{B}} \left\langle \Phi^{+} \right| (U_{A} \otimes U_{B}) \rho (U_{A}  \otimes U_{B})^{\dagger} \left|\Phi^{+} \right\rangle\ ,
\end{equation}
which is called the \emph{fully entangled fraction} \cite{fidelityCite} of the two-qubit state $\rho$ and measures the closeness of $\rho$ to a maximally entangled state.\newline

Theorem~\ref{nonDistabilittyTheorem} below states that the fully entangled fraction of the following class of states cannot 
be increased non-interactively. We call them  \emph{weakly entangled mixed states}, denoted $\Omega_{\alpha}$.
\begin{definition}
\label{tauDef}
{\rm 
The \emph{weakly entangled mixed state} $\Omega_{\alpha}$ with $\alpha \in [0,1]$ is defined as 
\begin{eqnarray}
  \Omega_{\alpha} &=& \alpha \left| \Phi^{+} \right\rangle \left\langle \Phi^{+} \right| + \frac{1-\alpha}{2} \left|00\right\rangle\left\langle 00\right| + \frac{1-\alpha}{2} \left|11\right\rangle \left\langle 11\right| \nonumber \\
     \label{omegaEqu1}        &=& \frac{1+\alpha}{2}  \left| \Phi^{+} \right\rangle \left\langle \Phi^{+} \right| + \frac{1-\alpha}{2}  \left| \Phi^{-} \right\rangle \left\langle \Phi^{-} \right| ~.\nonumber
\end{eqnarray}
}
\end{definition}

We need  some properties of weakly entangled mixed states. The proofs of Lemmas~\ref{wernerStateSimulation}  
and~\ref{weakEntangledMaximallCHSH} can be found in the appendix. 
First, we show which isotropic $q$-PR-boxes can be simulated using $\Omega_{\alpha}$.
\begin{lemma}
\label{wernerStateSimulation}
An isotropic $q$-PR-box can be simulated by  $\Omega_{\alpha}$, using appropriate measurements, if
\[
  q \leq \frac{1}{2} + \frac{1}{4} \sqrt{1 + \alpha^{2}} ~.
\]
\end{lemma}
Note that the simulation of boxes by measurements on entangled states fits nicely into the formalism of boxes described by terminals (see Section \ref{NDP}): A terminal corresponds to a local measurement performed on an entangled state.

\begin{lemma}
\label{weakEntangledMaximallCHSH}
Let $\Omega_{\alpha}$ be a weakly entangled mixed state and $\sigma^{ij}$ be arbitrary two-qubit quantum states and $\sum_{ij} p_{ij} = 1$ with $p_{ij} \geq 0$, then
\[
 F(\Omega_\alpha) \geq \sum_{ij} p_{ij} \cdot F(\sigma^{ij}) \Rightarrow NL[\Omega_{\alpha}] + c(\alpha) \geq \sum_{ij} p_{ij} \cdot NL[\sigma^{ij}] ~,
\]
where $c(\alpha) = \alpha (2 \sqrt{2} - 2) - 2 \sqrt{1 + \alpha^2} + 2$.
\end{lemma}

\subsection{Impossibility of Entanglement as well as Non-Locality Distillation}

In the following, we  consider {\em Entanglement Distillation Protocols\/} (EDP) \emph{without} communication \cite{entanglementDistillationUpperbound}. An EDP $\mathcal{E}$ takes as input 
 $\rho^{\otimes n}$, where $\rho$ is a (mixed) two-qubit state shared between Alice and Bob. The output of an EDP is one two-qubit (mixed) state $\sigma$. The operation of an EDP is denoted by $\mathcal{E}(\rho^{\otimes n}) = \sigma$. The performance of the EDP is measured by the fully entangled fraction-difference between the output state $\sigma$ and the input state $\rho$.
Without loss of generality, we can assume that the EDP consists of Alice and Bob applying local unitary operations $U_A$ and $U_B$ on their input quantum state $\rho_{AB}^{input}$ (ancilla qubits and $\rho^{\otimes n}$). In other words, we have $\mathcal{E}(\rho^{\otimes n}) = \Tr_{\bar{A}\bar{B}} \left( (U_A \otimes U_B) \rho_{AB}^{input} (U_A \otimes U_B)^\dagger \right)$ where $\Tr_{\bar{A}\bar{B}}$ means that Alice and Bob trace out everything but a single qubit each (this resulting two-qubit state is denoted by $\sigma$). Then, we can prove the following theorem which is mainly based on a result of Ambainis and Yang \cite{entanglementDistillationUpperbound}. It says that it is impossible to distill (without communication) out of many $\Omega_\alpha$ states a state which contains more entanglement than $\Omega_\alpha$.

\begin{theorem}
\label{nonDistabilittyTheorem}
For any probabilistic non-interactive entanglement-distillation protocol $\mathcal{E}$, the upper bound
\[
  F(\mathcal{E}(\Omega_{\alpha}^{\otimes n})) \leq \frac{1 + \alpha}{2} ~
\]
holds and is tight.
\end{theorem}

\begin{proof}
Assume  we have $n$ copies of the state $\left| \Phi^{+} \right\rangle$. Then, we perform a measurement in the computational basis on $r$ randomly chosen two-qubit states of $\left| \Phi^{+} \right\rangle^{\otimes n}$. Hence, for the measured states we get the result $\left|00\right\rangle$ and $\left|11\right\rangle$ with probability $1/2$ each, and the remaining $n - r$ two-qubit states stay unchanged.
We denote the resulting set of possible states by $\{ \left|\phi_{\vec{v}}^{r} \right\rangle \}_{\vec{v}} = \{ \bigotimes_{j=1}^{n} \left|\phi_{j}\right\rangle$ \}, where
\begin{displaymath}
\left|\phi_{j}\right\rangle = \left\{ \begin{array}{ll} \left|00\right\rangle & \textnormal{if } v[j] = 0  \\ \left|11\right\rangle & \textnormal{if } v[j] = 1 \\ \left|\Phi^{+}\right\rangle & \textnormal{if } v[j] = \ast  \end{array} \right.
\end{displaymath}
and $\vec{v} \in \{0,1,\ast \}^{n}$. The degree of $\vec{v}$, denoted by $\deg(\vec{v}) = r$, is the number of $j$'s in
$\{1,...,n\}$ for which $v[j] \neq \ast$. It is easy to see that there are $\genfrac(){0 pt}{}{n}{r} 2^{r}$ states in the set $\{ \left|\phi_{\vec{v}}^{r} \right\rangle \}_{\vec{v}}$ of degree $r$ which have all the same probability. We claim that the state $\Omega_{\alpha}^{\otimes n}$, written as 
\begin{equation}
\label{defNewOmega}
\Omega_{\alpha}^{\otimes n} = \left( \alpha \left| \Phi^{+} \right\rangle \left\langle \Phi^{+} \right| + \frac{1-\alpha}{2} \left( \left|00\right\rangle\left\langle 00\right| + \left|11\right\rangle \left\langle 11\right| \right) \right)^{\otimes n} ~,
\end{equation}
is just a mixture of the states $\omega_{r}^{n} \in \{ \left|\phi_{\vec{v}}^{r} \right\rangle \}_{\vec{v}}$, weighted with appropriate chosen probabilities, i.e.,
\begin{equation}
\label{sumOmega}
  \Omega_{\alpha}^{\otimes n} = \sum_{r = 0}^{n} \genfrac(){0 pt}{}{n}{r} \alpha^{n-r} (1-\alpha)^{r} \cdot \omega_{r}^{n}\ ,
\end{equation}
where $\omega_{r}^{n}$ is chosen uniformly at random from the set $\{ \left|\phi_{\vec{v}}^{r} \right\rangle \}_{\vec{v}}$.
This is indeed the case, because the quantum state $\omega_{r}^{n}$ can be written as an equal mixture of all possible states of degree $r$, i.e.,
\begin{equation}
  \omega_{r}^{n} = \frac{1}{2^{r}\genfrac(){0 pt}{}{n}{r}} \left( \sum_{\vec{v}:\deg(\vec{v}) = r} \left|\phi_{\vec{v}}^{r} \right\rangle \left\langle \phi_{\vec{v}}^{r} \right| \right)\ ,
\end{equation}
and be put into equation (\ref{sumOmega}), which yields 
\begin{equation}
\label{sumChaOmega}
  \Omega_{\alpha}^{\otimes n} = \sum_{r=0}^{n} \alpha^{n-r} \left( \frac{1-\alpha}{2} \right)^{r} \left( \sum_{\vec{v}:\deg(\vec{v}) = r} \left| \phi_{\vec{v}}^{r} \right\rangle \left\langle \phi_{\vec{v}}^{r} \right| \right) ~.
\end{equation}
Comparing (\ref{sumChaOmega}) and (\ref{defNewOmega}), one can easily see that  (\ref{sumOmega}) must be correct.

Equation (\ref{sumOmega}) implies that it is enough to analyze the behavior of the EDP $\mathcal{E}$ for the input state $\omega_{r}^{n}$. Ambainis and Yang have shown in \cite{entanglementDistillationUpperbound}  that 
\begin{equation}
\label{newProtocol}
  F^{+}(\mathcal{E}(\omega_{r}^{n})) \leq 1 - \frac{r}{2n}\ ,
\end{equation}
where $F^{+}(\cdot)$ is defined by
$F^{+}(\sigma) =  \left\langle \Phi^{+} \right| \sigma \left|\Phi^{+} \right\rangle$.
It is not difficult to show that there exists another EDP, denoted by $\mathcal{E}^{+}$, for any input $\rho^{\otimes n}$,  such that $F(\mathcal{E}(\rho^{\otimes n})) = F^{+}(\mathcal{E}^{+}(\rho^{\otimes n}))$. Therefore, equation  (\ref{newProtocol}) is equivalent to $F(\mathcal{E}(\omega_{r}^{n})) \leq 1 - \frac{r}{2n}$. Together with  (\ref{sumOmega}), we get
\begin{equation}
	F(\mathcal{E}(\Omega_{\alpha}^{\otimes n})) \leq \sum_{r = 0}^{n} \genfrac(){0 pt}{}{n}{r} \alpha^{n-r} (1-\alpha)^{r} \cdot \left( 1 - \frac{r}{2n} \right) = \frac{1 + \alpha}{2}\ ,
\end{equation}
which proves the statement. Taking the first of the $n$ two-qubit states proves tightness of our bound (see Proposition \ref{weakEntangledFidelity} in the appendix).
\end{proof}
Note that it is crucial for this proof that the entangled input states are $\Omega_{\alpha}$-states. If, for example, the input states were \emph{pure}, there would exist an entanglement distillation protocol without communication \cite{EntanglementDistillationPossible}, which would actually increase the fully entangled fraction. The states  $\Omega_{\alpha}$ are, actually,  the first class of two-qubit states for which it is provably impossible to distill entanglement without communication. It is an interesting open problem  to find other states with this property.\newline

Before we  prove our main result, we  state a well-established result about the simultaneous block diagonalization of the projectors corresponding to two projective measurements with binary outputs \cite{MasanesBlockDiagonal}.
\begin{lemma}
\label{simultDiagonal}
Let $P_0,P_1,Q_0,Q_1$ be four projectors acting on a Hilbert-space $\mathcal{H}$ such that $P_0 + P_1 = I$ and $Q_0 + Q_1 = I$. There exists an orthonormal basis in $\mathcal{H}$ where the four projectors $P_0,P_1,Q_0,Q_1$ are simultaneously block diagonal, in blocks of size $1 \times 1$ or $2 \times 2$. 
\end{lemma}

\noindent
For completeness reasons we provide a proof of Lemma \ref{simultDiagonal} in the appendix (taken from \cite{MasanesBlockDiagonal}).

\begin{theorem}
\label{resultGeneralDistillation}
There exists no non-locality distillation protocol $\mathcal{N}^{p}_{q}$, where the $q$-PR-boxes must be isotropic, for any $q \in [\frac{3}{4}, \frac{2 + \sqrt{2}}{4}]$ and $p > q + g(q)$ with
\[
 g(q) = \frac{3}{4} - q +  \frac{\sqrt{2} - 1}{4} \sqrt{4 (2 q - 1)^2 - 1} ~.
\]
\end{theorem}

Figure \ref{maximalIncrease} shows which $p$-PR-box could at most be distilled out of $q$-PR-boxes. Note that $g(3/4) = 0$ and $g((2 + \sqrt{2})/4) = 0$.

\begin{proof}
The proof is by contradiction. Assume there exists an NDP which uses isotropic $q$-PR-boxes and has the same input/output behavior as a $p$-PR-box such that $p > q + g(q)$ for some $q \in [\frac{3}{4}, \frac{2 + \sqrt{2}}{4}]$.
As argued in Section \ref{NDP}, every NDP can be represented by local operations and terminals on Alice's and Bob's side, respectively. Let $P(a,b | x,y)$ denote the input/output behavior of the NDP.

The local operations can be represented by local unitary operations and the terminals can be simulated by appropriate measurements on $\Omega_{\alpha}$ states (see Lemma \ref{wernerStateSimulation}). It is always possible to simulate the terminals, because the input boxes are isotropic $q$-PR-boxes with $q \in [\frac{3}{4}, \frac{2 + \sqrt{2}}{4}]$.

By introducing ancilla qubits, all intermediate measurements can be postponed to the end of the quantum circuits. Therefore, the quantum circuits can be represented by unitary operations $\tilde{U}_A^x$ and $\tilde{U}_B^y$, followed by measurements on the resulting  states, where $x$ and $y$ are Alice's and Bob's input bits to the NDP. As the output of the NDP is one bit on Alice's and Bob's site, respectively, the measurements after the unitary operation on Alice's (Bob's) site can each be represented by two \emph{positive operator-valued measurements} (POVMs) with \emph{binary} output. We need two POVMs on each side because the measurements can depend on the input bits $x$ and $y$, respectively. By introducing ancilla qubits, the POVMs can be replaced by projective measurements with binary outputs (Naimark's dilation theorem). 
We denote the resulting measurements by $\{\tilde{P}_A^{x,a}\}_a$ and $\{\tilde{P}_B^{y,b}\}_b$, respectively ($a$ and $b$ denote Alice's and Bob's output bit of the NDP). By merging the unitary operations $\tilde{U}_A^x$, $\tilde{U}_B^y$ and the measurements $\{\tilde{P}_A^{x,a}\}_a$, $\{\tilde{P}_B^{y,b}\}_b$ one gets the measurements
\begin{eqnarray}
\label{measurementA}
\{\hat{P}_A^{x,a}\}_a &=& \{ \tilde{U}_A^x \cdot \tilde{P}_A^{x,a} \cdot (\tilde{U}_A^x)^\dagger \}_a ~,~ x \in \{0,1\} \\
\{\hat{P}_B^{y,b}\}_b &=& \{ \tilde{U}_B^y \cdot \tilde{P}_B^{y,b} \cdot (\tilde{U}_B^y)^\dagger \}_b ~,~ y \in \{0,1\}
\end{eqnarray} 
applied on the input qubits (i.e., the $\Omega_\alpha$ and ancilla qubits).

Lemma \ref{simultDiagonal} tells us that there exists a basis such that the measurement operators of equation (\ref{measurementA}) can all be written in the \emph{same} block diagonal form, with blocks of size at most $2 \times 2$. Let us denote these transformed measurements by $\{P_A^{x,a}\}_a$, with $x \in \{0,1\}$, and the corresponding unitary basis change from the standard to the block diagonal basis by $U_A^{diag}$. Note that $U_A^{diag}$ does \emph{not} depend on Alice's input bit $x$ to the NDP because it is the same basis transformation for all four projectors. The same procedure can be applied on Bob's side, yielding block diagonal measurement operators $\{P_B^{y,b}\}_b$, with $y \in \{0,1\}$, and unitary operation $U_B^{diag}$. 

Let $\{ \bar{E}_A^i \}_i$ be a POVM where $\bar{E}_A^i$ is a projector onto the subspace of the simultaneous $i$th diagonal block of the four projectors $P_A^{x,a}$ with $x,a \in \{0, 1\}$. If $\mathcal{H}_A$ denotes the whole Hilbert space on Alice's side and $\mathcal{H}_A^i$ is the two dimensional subspace of $\mathcal{H}_A$ on which $\bar{E}_A^i$ projects, we can write the space $\mathcal{H}_A$ as a direct sum of orthogonal subspaces, i.e., $\mathcal{H}_A = \bigoplus_i \mathcal{H}_A^i$.  Furthermore, the projector $\bar{E}_A^i$ which acts on the whole space $\mathcal{H}_A$ can be decomposed as
\begin{equation}
\bar{E}_A^i = 0 \oplus 0 \oplus ... \oplus 0 \oplus E_A^i \oplus 0 \oplus ... \oplus 0 ~,
\end{equation}
where $E_A^i : \mathcal{H}_A^i \rightarrow \mathcal{H}_A^i$ is at the $i$th position in the direct sum and is actually just the identity on $\mathcal{H}_A^i$. Similarly for Bob.

If we measure the POVMs $\{ \bar{E}_A^i \}_i$ and $\{ \bar{E}_B^j \}_j$ on Alice's and Bob's side, respectively, after the unitary operations $U_A^{diag}$ and $U_B^{diag}$ but before the measurements $\{P_A^{x,a}\}_a$ and $\{P_B^{y,b}\}_b$, the outcome statistic of Alice and Bob does not change. In order to see this, we compute the corresponding input/output statistic of Alice and Bob. 
Let us denote the state of Alice and Bob right after the unitaries $U_A^{diag}$ and $U_B^{diag}$ by $\rho_{AB}$. The mixed state after the measurements $\{ \bar{E}_A^i \}_i$ and $\{ \bar{E}_B^j \}_j$ is then given by
\begin{equation}
\label{eqnStateAfterMeas}
 \tilde{\rho}_{AB} = \sum_i \sum_j (\bar{E}_A^i \otimes \bar{E}_B^j) \cdot \rho_{AB} \cdot (\bar{E}_A^i \otimes \bar{E}_B^j) = \sum_{ij} p_{ij} \cdot \bar{\rho}_{AB}^{ij}  ~,
\end{equation}
with $\tilde{\rho}_{AB}$ on $\mathcal{H}_A \otimes \mathcal{H}_B$ and $p_{ij} = \Tr((\bar{E}^i_A \otimes \bar{E}^j_B) \rho_{AB})$ and where $\bar{\rho}_{AB}^{ij} = 0 \oplus 0 \oplus ... \oplus 0 \oplus \rho_{AB}^{ij} \oplus 0 \oplus ... \oplus 0$ are the normalized post-measurement states with $\rho_{AB}^{ij}$ on $\mathcal{H}_A^i \otimes \mathcal{H}_B^j$, i.e., $\bar{\rho}_{AB}^{ij}$ is just an embedding of $\rho_{AB}^{ij}$ into $\mathcal{H}_A \otimes \mathcal{H}_B$.
Using the cyclicity of the trace and the fact that 
\begin{equation}
 P_A^{x,a} = \sum_i \bar{E}_A^i \cdot P_A^{x,a} \cdot \bar{E}_A^i ~~ \forall x,a \in \{0,1\} ~,~ P_B^{y,b} = \sum_j \bar{E}_B^j \cdot P_B^{y,b} \cdot \bar{E}_B^j ~~ \forall y,b \in \{0,1\}\ ,
\end{equation}
one can easily show that:
\begin{equation}
\label{eqnInputOutput}
 P(a,b \vert x,y) \equiv \Tr \left( (P_A^{x,a} \otimes P_B^{y,b}) \cdot \rho_{AB} \right) = \Tr \left( (P_A^{x,a} \otimes P_B^{y,b}) \cdot \tilde{\rho}_{AB} \right) ~~ \forall x,y,a,b \in \{0,1\}~.
\end{equation}
Therefore, we do not change the input/output statistic by these intermediate measurements.

In the following we will try to map all information onto just two qubits so that we can apply Theorem \ref{nonDistabilittyTheorem} which says the the fully entangled fraction of any two qubit output state of any entanglement distillation protocol does not have a larger fully entangled fraction than the input states, if the input states are $\Omega_\alpha$'s. We will do this mapping in a clever way such that the fully entangled fraction of this two qubit state is equal to the convex combination of the fully entangled fractions of the states in the convex decomposition (see equation (\ref{eqnFullyEn})). In order to achieve this we will proceed as follows.

Let us define new projective measurement operators on the individual subspaces by
\begin{eqnarray}
\bar{Q}_{A}^{i,x,a} &:=& \bar{E}_A^i \cdot P_A^{x,a} \cdot \bar{E}_A^i ~, \\
\bar{Q}_{B}^{j,y,b} &:=& \bar{E}_B^j \cdot P_B^{y,b} \cdot \bar{E}_B^j ~,
\end{eqnarray}
for all $x,y,a,b \in \{0,1\}$ and $\bar{Q}_{A}^{i,x,a} = 0 \oplus 0 \oplus ... \oplus 0 \oplus Q_{A}^{i,x,a} \oplus 0 \oplus ... \oplus 0$ with $Q_{A}^{i,x,a}$ acting on $\mathcal{H}_A^i$. Then, by using equations (\ref{eqnStateAfterMeas}) and (\ref{eqnInputOutput}), the cyclicity of the trace and the idempotence of projectors,  we get
\begin{eqnarray}
P(a,b \arrowvert x,y ) &=& \Tr \left( (P_A^{x,a} \otimes P_B^{y,b}) \cdot \tilde{\rho}_{AB} \right) \nonumber \\
&=& \sum_{ij} p_{ij} \cdot \Tr\left( (P_A^{x,a} \otimes P_B^{y,b}) (\bar{E}_A^i \otimes \bar{E}_B^j) \bar{\rho}_{AB}^{ij} (\bar{E}_A^i \otimes \bar{E}_B^j) \right) \nonumber \\
&=& \sum_{ij} p_{ij} \cdot \Tr \left( (\bar{Q}_{A}^{i,x,a} \otimes \bar{Q}_B^{j,y,b}) \cdot \bar{\rho}_{AB}^{ij} \right) \nonumber \\
&=& \sum_{ij} p_{ij} \cdot \Tr \left( (Q_{A}^{i,x,a} \otimes Q_B^{j,y,b}) \cdot \rho_{AB}^{ij} \right) \nonumber \\
&=& \sum_{ij} p_{ij} \cdot P_{\rho_{AB}^{ij}}(a,b \arrowvert x,y ) ~,
\end{eqnarray}
where $P_{\rho_{AB}^{ij}}(a,b \arrowvert x,y )$ denotes the input/output statistic resulting from measuring the state $\rho_{AB}^{ij}$ by the measurement operators $\{ Q_{A}^{i,x,a} \otimes Q_B^{j,y,b} \}_{a,b}$. 

By using Lemma \ref{weakEntangledMaximallCHSH} and the fact that mixing input/output statistics cannot increase the non-locality we get
\begin{eqnarray}
\label{eqnNLupper}
NL[P(a,b \arrowvert x,y)] &\leq& \sum_{ij} p_{ij} \cdot NL[P_{\rho_{AB}^{ij}}(a,b \arrowvert x,y)] \nonumber \\
&\leq& \sum_{ij} p_{ij} \cdot NL[\rho_{AB}^{ij}] \nonumber \\
&\leq& NL[\Omega_\alpha] + c(\alpha) ~.
\end{eqnarray}
In order to be allowed to apply Lemma \ref{weakEntangledMaximallCHSH} in the last inequality of equation (\ref{eqnNLupper}) we need to prove that 
\begin{equation}
\label{eqnFidelityUpper}
F(\Omega_\alpha) \geq \sum_{ij} p_{ij} \cdot F(\rho_{AB}^{ij}) ~.
\end{equation}
This can be shown as follows. Let us introduce a completely positive map (CPM) $\mathcal{M} : \mathcal{H}_A \otimes \mathcal{H}_B \rightarrow \mathcal{H}_2^A \otimes \mathcal{H}_2^B$, where $\mathcal{H}_2^A$ and $\mathcal{H}_2^B$ are single qubit Hilbert spaces, which we define by $\mathcal{M}(\tilde{\rho}_{AB}) = \sum_{ij} (\bar{U}_A^i \otimes \bar{U}_B^j) \tilde{\rho}_{AB} (\bar{U}_A^i \otimes \bar{U}_B^j)^\dagger$, where the operators $\bar{U}_A^i : \mathcal{H}_A \rightarrow \mathcal{H}_2^A$ and $\bar{U}_B^j : \mathcal{H}_B \rightarrow \mathcal{H}_2^B$ are defined as
\begin{eqnarray}
\bar{U}_A^i &=& 0 \oplus 0 \oplus ... \oplus 0 \oplus U_A^i \oplus 0 \oplus ... \oplus 0 ~, \\
\bar{U}_B^j &=& 0 \oplus 0 \oplus ... \oplus 0 \oplus U_B^j \oplus 0 \oplus ... \oplus 0 ~,
\end{eqnarray}
where $U_A^i : \mathcal{H}_A^i \rightarrow \mathcal{H}_2^A$ and $U_B^j : \mathcal{H}_B^j \rightarrow \mathcal{H}_2^B$ are unitary operations on single qubit subspaces on which $E_A^i$ and $E_B^j$ project, respectively (the CPM $\mathcal{M}$ is actually a trace preserving operations as $\sum_{ij} (\bar{U}_A^i \otimes \bar{U}_B^j)^\dagger \cdot (\bar{U}_A^i \otimes \bar{U}_B^j) = I$). How we choose these unitary operations will be explained later.  As $\tilde{\rho}_{AB}$ is the state after the measurements $\{ \bar{E}_A^i \}_i$ and $\{ \bar{E}_B^j \}_j$ it has block diagonal structure and therefore, if we apply the CPM $\mathcal{M}$ to it, we get the following two qubit state
\begin{eqnarray}
\hat{\rho}_{AB} = \mathcal{M}(\tilde{\rho}_{AB}) &=& \sum_{ij}  (\bar{U}_A^i \otimes \bar{U}_B^j) \cdot \tilde{\rho}_{AB} \cdot (\bar{U}_A^i \otimes \bar{U}_B^j)^\dagger \nonumber \\ 
&=& \sum_{ij} p_{ij} \cdot (U_A^i \otimes U_B^j) \cdot \rho_{AB}^{ij} \cdot (U_A^i \otimes U_B^j)^\dagger ~.
\end{eqnarray}
Computing the fully entangled fraction of $\hat{\rho}_{AB}$ yields
\begin{eqnarray}
\label{eqnFullyEn}
F(\hat{\rho}_{AB}) &=& \max_{U_A,U_B} \langle \Phi^+ | (U_A \otimes U_B) \hat{\rho}_{AB} (U_A \otimes U_B)^\dagger | \Phi^+ \rangle \nonumber \\
&=&   \sum_{ij} p_{ij} \cdot \max_{U_A,U_B} \langle \Phi^+ | (U_A \otimes U_B) (U_A^i \otimes U_B^j) \rho_{AB}^{ij} (U_A^i \otimes U_B^j)^\dagger (U_A \otimes U_B)^\dagger | \Phi^+ \rangle \nonumber \\
&=& \sum_{ij} p_{ij} \cdot F(\rho_{AB}^{ij}) ~,
\end{eqnarray}
where we have chosen the $U_A^i$'s and $U_B^j$'s in such a way that they maximize the fidelity for each $\rho_{AB}^{ij}$  and we therefore get the fully entangled fraction $F(\rho_{AB}^{ij})$ for each two qubit state $\rho_{AB}^{ij}$. Furthermore, as the CPM $\mathcal{M}$ can be replaced by a unitary operation by introducing ancilla qubits we know, due to Theorem \ref{nonDistabilittyTheorem}, that the fully entangled fraction of $\hat{\rho}_{AB}$ is smaller or equal than the fully entangled fraction of $\Omega_\alpha$, i.e., $F(\Omega_\alpha) \geq F(\hat{\rho}_{AB})$. Together with equation (\ref{eqnFullyEn}) this argument finishes the proof of equation (\ref{eqnFidelityUpper}).

Hence, we have proved equation (\ref{eqnFidelityUpper}) and therefore also equation (\ref{eqnNLupper}). This implies that the non-locality of the NDP, denoted by $NL[P(a,b \arrowvert x,y)]$, is smaller or equal than the non-locality of the input state plus some term $c(\alpha)$. By using the the fact that $NL[\Omega_\alpha] = 2 \sqrt{1 + \alpha^2}$ (see Proposition \ref{weakEntangledCHSH} in the appendix) and the relation $q = 1/2 + 1/8 \cdot NL[\Omega_\alpha]$, which is a rescaling of the non-locality from the interval $[2, 2\sqrt{2}]$ to $[3/4,(2 + \sqrt{2})/4]$, we can re-express equation (\ref{eqnNLupper}) by $p \leq q + g(q)$.
But this is in contradiction to our initial assumption, which demanded that the NDP has an input/output statistic which is more non-local than $q + g(q)$. 
\end{proof}

Figure \ref{maximalIncrease} shows the limited distillability of $q$-PR boxes. The largest gap can be found at $q = \frac{1}{2} + \frac{1}{4} \sqrt{\frac{1 + \sqrt{2}}{2}} \approx 0.775$ which results in a gap of around $g(0.775) \approx 0.0225$. Note that the only possible consistent way of how non-locality could be distilled must have step function form. First, it is clear that out of $(q + \epsilon)$-PR-boxes (with $\epsilon > 0$) one can always distill as much non-locality as of $q$-PR-boxes. Second, assume that we can at most distill a $(q + \epsilon)$-PR-box out of $q$-PR-boxes. If we now could distill a $(q + \epsilon + \gamma)$-PR-box out of $m$ $(q + \delta)$-PR-boxes, with $\delta < \epsilon$, and could distill a $(q + \delta)$-PR-box out of $n$ $q$-PR-boxes, we could simulate a $(q + \epsilon + \gamma)$-PR-box with $n \cdot m$ $q$-PR-boxes. But this is a contradiction to our assumption that at most a $(q + \epsilon)$-PR-box can be distilled out of $q$-PR-boxes.

Furthermore, we get an immediate corollary out of Theorem \ref{resultGeneralDistillation}, namely:

\begin{figure}
 \centering
 \input{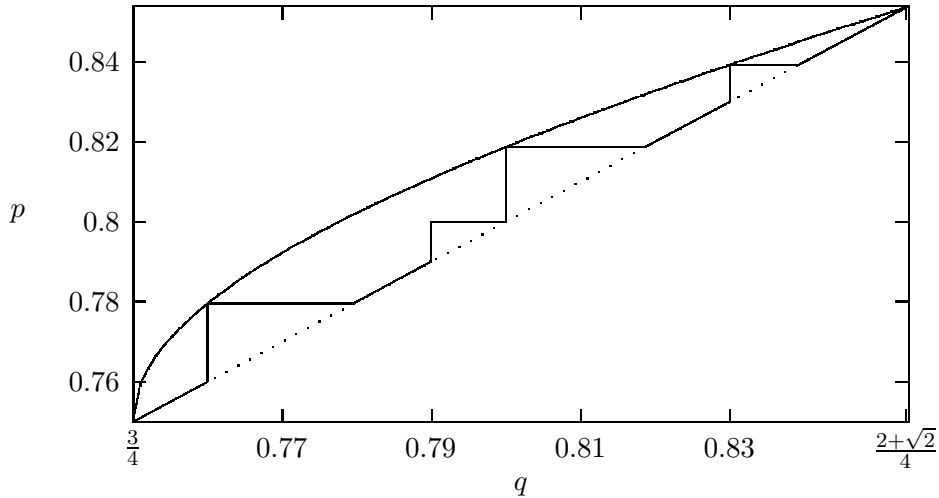}  
 \caption{The maximal amount of distillable non-locality is shown. For given isotropic $q$-PR-boxes at most a $p$-PR-box which has a $p$ below the top curve $q + g(q)$ could be distilled. The step function plotted between the upper limit curve and the dotted non-distillability line shows one possible behavior of how much non-locality could be distilled for given $q$-PR-boxes.}
\label{maximalIncrease}
\end{figure}

\begin{corollary}
There are infinitely many $q \in [3/4, (2 + \sqrt{2})4]$ for which it is impossible to distill a $p$-PR-box out of isotropic $q$-PR-boxes for \emph{any} $p > q$. 
\end{corollary}

\begin{proof}
That there are infinitely many points where non-locality distillation is impossible follows from the fact that any non-locality distillation protocol has the above described step function feature. This is the case because the curve $q + g(q)$ and the line $q$ are getting closer and closer for larger $q$ and finally meet at the point $(2 + \sqrt{2})/4$. 
\end{proof}

\section{Concluding Remarks}
\label{sectionConclusion}
The joint behavior of the two parts of a bipartite input-output system is {\em non-local\/} 
if it cannot be explained by shared (classical) information. The interest of non-locality 
comes from two facts: First, it is {\em real}, i.e., appears as the measurement-outcome 
behavior of certain entangled quantum states | even if the measurement events are space-like 
separated. This does not formally contradict relativity as long as the system is {\em non-signaling},
i.e., does not allow for message transmission. However, it contradicts the {\em spirit\/}  of 
relativity, which is  the reason why Einstein {\em et al.}~\cite{epr} had asked for local explanations.
Second, it is {\em useful\/} for information processing, e.g., for reducing the communication complexity 
or for achieving unconditional confidentiality. Since non-locality is more useful if it is stronger, 
non-locality amplification is a natural goal.

Only little is known on whether non-locality amplification was possible: 
Two bounds cannot be over-stepped, namely the 
Bell and Tsirelson bounds. On the other hand, {\em perfect non-locality\/} cannot be obtained
from imperfect. Here, we show the first result stating that only \emph{limited} distillability could be possible for a large class of systems, namely isotropic quantum-physically realizable approximations to 
PR-boxes. Our results imply that there are \emph{infinitely} many $q \in [3/4, (2 + \sqrt{2})/4]$ where non-locality distillation is impossible at all. 

It is somewhat remarkable that our result, which talks only about classical systems and circuits,
has a ``quantum-physical proof'': It makes (central) use of facts from quantum theory and information 
science. An example is a novel result stating that the fully entangled fraction of a class of mixed entangled states cannot be increased 
non-interactively at all. This is the first result  of this strength, and it is a challenging open problem
to find other  states with this property. 

Our result complements~\cite{FosterBox}, where it is shown that certain non-isotropic non-local systems 
{\em can} be distilled. This may suggest that the amount of CHSH violation is a good measure for non-locality 
only for isotropic systems. 
We state as an open problem to find a general proof  that 
isotropic non-locality can {\em never\/} be amplified.

\ \\
\noindent
{\bf Acknowledgments.}
We thank Nicolas Gisin, Renato Renner and J\"urg Wullschleger 
for many helpful discussions, and an anonymous referee for pointing out to us the 
reference about simultaneous block-diagonalization of two bipartite measurements which we could apply in order to correct an earlier version of this paper.

\newpage

\appendix

\renewcommand{\theequation}{A-\arabic{equation}}
\setcounter{equation}{0}  
\section*{APPENDIX}  

\setcounter{lemma}{1}

\section{Proof of Lemma \ref{corollarryQuantumBound}}

\begin{lemma}{\bf (Tsirelson \cite{TsirelsonBound})}
There exists no non-locality distillation protocol $\mathcal{N}^{p}_{q}$, where the $q$-PR-boxes must be quantum-physically realizable, with $q \leq \frac{2 + \sqrt{2}}{4}$ and $p > \frac{2 + \sqrt{2}}{4}$. 
\end{lemma}

\begin{proof}
The proof is by contradiction. Assume there exists an NDP which uses quantum-physically obtainable $q$-PR-boxes and has the same input/output behavior as a $p$-PR-box such that $p > \frac{2 + \sqrt{2}}{4}$.
By assumption, the behavior of the $q$-PR-boxes can be obtained by measuring some entangled states. 
It can be shown that every such circuit which consists of classical local operations, shared randomness,
 and measurement operators on entangled quantum states is equivalent to some Hermitian measurement operators $X_{A}^0$, $X_{A}^1$, $Y_{B}^0$, and $Y_{B}^1$ operating on these entangled quantum states and ancilla qubits.
The notation $X_{A}^i$ means that if Alice's input bit is $i$, she applies the observable $X_{A}^i$ on her part of the entangled and ancilla qubits; similarly for Bob. The observables $X_{A}^0$, $X_{A}^1$, $Y_{B}^0$, and $Y_{B}^1$ have binary output (but arbitrary finite input alphabet), which corresponds to the output of the NDP. Tsirelson has shown in \cite{TsirelsonBound} that for any observables of finite dimension with binary output which are applied on any entangled state, the maximal CHSH non-locality that can be achieved is $2 \sqrt{2}$, which corresponds to a $\frac{2 + \sqrt{2}}{4}$-PR-box. Therefore, the NDP we have constructed, and transformed into Hermitian measurement operators on entangled quantum states, must also obey this bound. 
\end{proof}

\renewcommand{\theequation}{B-\arabic{equation}}
\setcounter{equation}{0}  
\setcounter{lemma}{3}

\section{Proof of Lemma \ref{wernerStateSimulation}}

\begin{proposition}
\label{weakEntangledCHSH}
The non-locality of $\Omega_{\alpha}$  is 
\[
NL[\Omega_{\alpha}] = 2 \sqrt{1 + \alpha^{2}} ~.
\]
\end{proposition}

\begin{proof}
Assume we have given a Bell diagonal mixed quantum state
\begin{equation}
\rho = p_{1} \left|\Phi^{+}\right\rangle\left\langle \Phi^{+}\right| + p_{2} \left|\Phi^{-}\right\rangle\left\langle \Phi^{-}\right| + p_{3} \left|\Psi^{+}\right\rangle\left\langle \Psi^{+}\right| + p_{4} \left|\Psi^{-}\right\rangle\left\langle \Psi^{-}\right| ~.
\end{equation}
Using the results from \cite{violatingBellInequality} one gets for the maximal CHSH violation of $\rho$ the following value
\begin{equation}
NL[\rho] = 2 \sqrt{\lambda_{1} + \lambda_{2}}
\end{equation}
where $\lambda_{1}$ and $\lambda_{2}$ are the two larger values of the following three expressions
\begin{equation}
\label{eqnList}
(2 \cdot(p_{1} + p_{3}) - 1)^{2} ~,~ (2 \cdot(p_{2} + p_{3}) - 1)^{2} ~,~ (2 \cdot(p_{1} + p_{2}) - 1)^{2} ~.
\end{equation}
As for $\Omega_{\alpha}$ we have $p_{1} = (1+\alpha)/2$, $p_{2} = (1-\alpha)/2$ and $p_{3} = p_{4} = 0$, the two larger values of equation (\ref{eqnList}) for $\lambda_{1}$ and $\lambda_{2}$ are
\begin{eqnarray}
\lambda_{1} &=& \alpha^{2} ~, \nonumber\\
\lambda_{2} &=& 1 ~. \nonumber
\end{eqnarray} 
Hence, the maximal CHSH violation for $\Omega_{\alpha}$ becomes $NL[\Omega_{\alpha}] = 2 \sqrt{1 + \alpha^{2}}$.
\end{proof}

\begin{lemma}
An isotropic $q$-PR-box can be simulated by  $\Omega_{\alpha}$, using appropriate measurements, if
\[
  q \leq \frac{1}{2} + \frac{1}{4} \sqrt{1 + \alpha^{2}} ~.
\]
\end{lemma}

\begin{proof}
The  state $\Omega_{\alpha}$ has  non-locality of $NL[\Omega_{\alpha}] = 2 \sqrt{1 + \alpha^{2}}$ when  appropriate measurement operators are chosen by Alice and Bob (see Proposition \ref{weakEntangledCHSH}). We will see that it suffices to consider measurements in the real plane to achieve this non-locality. Therefore, such measurements can be described by observables of the form 
\begin{equation}
 M(\vec{v}) = \langle \vec{v} , \vec{\sigma} \rangle\ ,
\end{equation}
where $\vec{v} = (\cos(\phi), \sin(\phi))^T \in \mathbb{R}^2$ is a real unit vector (the measurement direction) and $\vec{\sigma} = (\sigma_X, \sigma_Z)^T$, where $\sigma_X$ and $\sigma_Z$ are Pauli matrices.
Let us define the following observables for Alice: 
\begin{equation}
\label{measureMaxA}
M(\vec{a}_0) = \sigma_Z  ~,~ M(\vec{a}_1) = \sigma_X 
\end{equation}
and for Bob,
\begin{equation}
\label{measureMaxB}
M(\vec{b}_0) = \frac{\alpha}{\sqrt{1 + \alpha^2}} \sigma_X + \frac{1}{\sqrt{1 + \alpha^2}} \sigma_Z ~,~ M(\vec{b}_1) =  \frac{-\alpha}{\sqrt{1 + \alpha^2}} \sigma_X + \frac{1}{\sqrt{1 + \alpha^2}} \sigma_Z ~.
\end{equation}
As these four Hermitian measurement operators have the eigenvalues $+1$ and $-1$, each of them can be written as the difference of two projectors, namely $M(\vec{a}_x) = P_A^{x,0} - P_A^{x,1}$ and $M(\vec{b}_y) = P_B^{y,0} - P_B^{y,1}$. Hence, the measurement statistic of the joint system of Alice and Bob becomes 
\begin{equation}
\label{probMeasu}
 P(a,b \vert x,y) = \Tr \left( (P_A^{x,a} \otimes P_B^{y,b}) \cdot \Omega_{\alpha} \right) ~~ \forall x,y,a,b \in \{0,1\}~.
\end{equation}
To compute the non-locality of $P(a,b \vert x,y)$ we need to determine the correlation functions (see equation (\ref{corrFunc})):
\begin{equation}
  E_{x,y} = P(0,0 \vert x,y) + P(1,1 \vert x,y) - P(0,1 \vert x,y) - P(1,0 \vert x,y) ~,~  x,y \in \{0,1\} ~.
\end{equation}
By using  (\ref{probMeasu}) and the linearity of the trace one gets
\begin{equation}
 E_{x,y} = \Tr \left( (M(\vec{a}_x) \otimes M(\vec{b}_y)) \cdot \Omega_{\alpha} \right) ~.
\end{equation}
By some easy calculations, one gets for the four correlation functions the following expressions
\begin{equation}
E_{0,0} = \frac{1}{\sqrt{1 + \alpha^2}} ~,~ E_{0,1} = \frac{1}{\sqrt{1 + \alpha^2}}  ~,~ E_{1,0} = \frac{\alpha^2}{\sqrt{1 + \alpha^2}}  ~,~ E_{1,1} = - \frac{\alpha^2}{\sqrt{1 + \alpha^2}} ~.
\end{equation}
This gives for the non-locality of $P(a,b \vert x,y)$ (see equation (\ref{CHSHeqn})), which corresponds to measuring the state $\Omega_\alpha$ by the measurement operators $M(\vec{a}_x)$ and $M(\vec{b}_y)$, for $x,y \in \{0,1\}$, a value of
\begin{equation}
 NL[P(a,b \vert x,y)] \equiv NL[\Omega_\alpha] = E_{0,0} + E_{0,1} + E_{1,0} - E_{1,1} = 2 \sqrt{1 + \alpha^2} ~.
\end{equation}
Lemma \ref{weakEntangledCHSH} tells us that this is the maximal possible value we can get for the non-locality of $\Omega_\alpha$. Hence, the measurements we defined in equations (\ref{measureMaxA}) and (\ref{measureMaxB}) extract as much non-locality out of $\Omega_\alpha$ as possible.

Using  $q = 1/2 + 1/8 \cdot NL[\Omega_{\alpha}]$, which is just a rescaling of CHSH-type non-locality from the interval $[2,2\sqrt{2}]$ to $[\frac{3}{4},\frac{2 + \sqrt{2}}{4}]$, we get the upper bound given in the lemma. But the joint input/output behavior $P(a,b \vert x,y)$ of Alice and Bob does not yet describe an \emph{isotropic} PR-box. 

Two requirements have to be fulfilled by $P(a,b \vert x,y)$ such that it corresponds to an isotropic PR-box. Namely, the local output by Alice and Bob must be unbiased, i.e., $P_A(a \vert x) = P_B(b \vert y) = 1/2$ for all $x, y, a, b \in \{0,1\}$, and the correlation functions must satisfy $E_{0,0} = E_{0,1} = E_{1,0} = - E_{1,1}$ (see Section \ref{sectionCHSH}). Computing the local input/output behavior of Alice (and similarly for Bob) yields
\begin{equation}
 P_A(a \vert x) = \Tr \left( P^{x,y}_A \cdot \Tr_B(\Omega_\alpha) \right) = \Tr \left( P^{x,a}_A \cdot I/2 \right) = \frac{1}{2} \Tr(P^{x,a}_A) = \frac{1}{2} ~~ \forall x,a \in \{0,1\}
\end{equation}
where $\Tr_B$ means taking the partial trace over Bob's system and $I$ is the identity matrix of dimension two. Therefore, the input/output behavior of Alice (Bob) is unbiased. In order to get the correct form for the correlation functions, Alice and Bob apply together\footnote{hence, they need one bit of shared randomness for this procedure} the following two operations, each with probability $1/2$:
\begin{enumerate}
 \item Alice and Bob do nothing, or
 \item Alice flips her input bit $x$ and Bob flips his output bit $b$ if his input bit $y$ is equal 1.
\end{enumerate}
It is not hard to show that Alice and Bob induce the following transformations on the four correlation functions by applying step 2.
\begin{equation}
 E_{0,0} \rightarrow E_{1,0} ~,~ E_{0,1} \rightarrow - E_{1,1} ~,~ E_{1,0} \rightarrow E_{0,0} ~,~ E_{1,1} \rightarrow - E_{0,1} ~.
\end{equation}
The resulting correlation functions are then (note that the local input/output behavior has not been changed by this process)
\begin{equation}
E_{0,0} = \frac{\sqrt{1 + \alpha^2}}{2} ~,~ E_{0,1} = \frac{\sqrt{1 + \alpha^2}}{2}  ~,~ E_{1,0} = \frac{\sqrt{1 + \alpha^2}}{2}  ~,~ E_{1,1} = - \frac{\sqrt{1 + \alpha^2}}{2} ~.
\end{equation}
As $E_{0,0} + E_{0,1} + E_{1,0} - E_{1,1}$ is still $2 \sqrt{1 + \alpha^2}$, we did not lose any non-locality by this process. Furthermore, we now have an isotropic PR-box, as desired.
\end{proof}

\renewcommand{\theequation}{C-\arabic{equation}}
\setcounter{equation}{0}  

\section{Proof of Lemma \ref{weakEntangledMaximallCHSH}}

\begin{proposition}
\label{weakEntangledFidelity}
The fully entangled fraction of  $\Omega_{\alpha}$ is 
\[
F(\Omega_{\alpha}) = \frac{1 + \alpha}{2} ~.
\]
\end{proposition}

\begin{proof}
In order to find $F(\Omega_{\alpha})$, we have to find unitary operations $U_{A}$ and $U_{B}$ which maximize $\left\langle \Phi^{+}\right|(U_{A} \otimes U_{B})^{\dagger} \Omega_{\alpha} (U_{A} \otimes U_{B}) \left|\Phi^{+}\right\rangle$. Let us define the joint unitary operation by
\begin{equation}
(U_{A} \otimes U_{B})\left|\Phi^{+}\right\rangle \equiv \left|\Psi\right\rangle= a_{00}\left|00\right\rangle + a_{01}\left|01\right\rangle + a_{10}\left|10\right\rangle + a_{11}\left|11\right\rangle
\end{equation} 
where $a_{00}$, $a_{01}$, $a_{10}$ and $a_{11}$ are complex numbers. Computing $\left\langle \Psi\right|\Omega_{\alpha}\left|\Psi\right\rangle$ yields the following equation after some easy calculations
\begin{equation}
\label{eqnMaxFid}
\left\langle \Psi\right|\Omega_{\alpha}\left|\Psi\right\rangle = \frac{1}{2} \cdot \left(|a_{00}|^{2} + |a_{11}|^{2} + \alpha \cdot (a_{00}^{*} \cdot a_{11} + a_{11}^{*} \cdot a_{00})\right) ~.
\end{equation} 
What is the maximal value this equation can attain? As $a_{01}$ and $a_{10}$ do not appear in the equation, they can be set to $0$. Therefore it must hold that $|a_{00}|^{2} + |a_{11}|^{2} = 1$ because of the unit length of quantum state vectors. 
As $|a_{00}|^{2} + |a_{11}|^{2} = 1$, it is easy to see that equation (\ref{eqnMaxFid}) attains its maximal value when $a_{00} = a_{11} = \frac{1}{\sqrt{2}}$. Hence, this equation becomes maximal when $\left|\Psi\right\rangle = \left|\Phi^{+}\right\rangle$ and therefore the fully entangled fraction of $\Omega_{\alpha}$ becomes
\begin{equation}
F(\Omega_{\alpha}) = \left\langle \Phi^{+}\right| \Omega_{\alpha} \left|\Phi^{+}\right\rangle = \frac{1+\alpha}{2} ~,
\end{equation}  
which finishes the proof.
\end{proof}

\begin{proposition}
\label{proMaxNonlocal}
Let $\rho$ be an arbitrary two qubit quantum state with $\frac{1}{2} \leq F(\rho) \leq 1$. Then the non-locality of $\rho$ is upper bounded by
\[
 NL[\rho] \leq 2 \sqrt{1 + (2 F(\rho) - 1)^2} ~.
\]
\end{proposition}

\begin{proof}
We will actually prove a stronger result, namely that $\Omega_\alpha$ is maximally non-local for given fully entangled fraction, i.e., $F(\Omega_{\alpha}) \geq F(\rho) \Rightarrow NL[\Omega_{\alpha}] \geq NL[\rho]$. If we can prove this result, the proposition follows, because $F(\Omega_{\alpha}) = \frac{1 + \alpha}{2}$ and $NL[\Omega_\alpha] = 2 \sqrt{1 + \alpha^2}$, due to Proposition \ref{weakEntangledCHSH} and \ref{weakEntangledFidelity}, respectively.

An arbitrary two qubit quantum state can be written as 
\begin{equation}
\rho = \frac{1}{4} \left( I \otimes I + \vec{r} \cdot \vec{\sigma} \otimes I + I \otimes \vec{s} \cdot \vec{\sigma} + \sum_{i,j = 1}^{3} T_{ij} \sigma_{i} \otimes \sigma_{j} \right)
\end{equation}
where $I$ stands for the identity operator, $\vec{r}, \vec{s} \in \mathbb{R}^{3}$ and $\vec{r} \cdot \vec{\sigma} = \sum_{i=1}^{3} r_{i} \sigma_{i}$ with $\{\sigma_{i}\}_{i=1}^{3}$ being the Pauli matrices and $T_{ij} \in [-1,1]$.

In the following we are mainly interested in the real $3\times 3 $ - matrix $T$ which is defined as $T_{ij} = \Tr(\rho ( \sigma_{i} \otimes \sigma_{j}))$, because the fully entangled fraction and the non-locality of $\rho$ depend only on the coefficients $T_{ij}$ as shown by the Horodecki family and Badzi\c{a}g in \cite{violatingBellInequality} and \cite{localEnvironmentFidelity}. As shown in \cite{localEnvironmentFidelity} one can find local unitary operations $U_{A}$ and $U_{B}$ such that $\rho_{d} = (U_{A} \otimes U_{B}) \rho (U_{A} \otimes U_{A})^{\dagger}$ has a diagonal matrix $T$, i.e., $T = diag(T_{11},T_{22},T_{33})$.

The non-locality of $\rho_{d}$, measured by the maximal violation of the CHSH inequality, is given by \cite{violatingBellInequality}
\begin{equation}
\label{nonlocalityEquation}
NL[\rho_{d}] = 2 \sqrt{T_{11}^{2} + T_{22}^{2}}
\end{equation}
where we assumed without loss of generality that $T_{11}$ and $T_{22}$ have the larger absolute value than $T_{33}$. In order for $\rho_d$ to be non-local, we need $T_{11}^{2} + T_{22}^{2} > 1$, which implies $\left|T_{11}\right| + \left|T_{22}\right| > 1$. By using the following two implications from \cite{localEnvironmentFidelity},
\begin{eqnarray}
\left|T_{11}\right| +  \left|T_{22}\right| +  \left|T_{33}\right| > 1 &\Rightarrow& F(\rho_d) > \frac{1}{2}\ , \\
F(\rho_d) > \frac{1}{2} &\Rightarrow& \det(T) < 0\ ,
\end{eqnarray}
we can conclude that $\det(T) = T_{11} \cdot T_{22} \cdot T_{33} < 0$ if $\rho_d$ is non-local.
Hence, because we are only interested in non-local states we must have $\det(T) < 0$. Furthermore, the unitary operations $U_{A}$ and $U_{B}$ can be chosen in such a way that $T_{11}, T_{22} \leq 0$. Together with $\det(T) < 0$ we can conclude that $T_{11},T_{22},T_{33} \leq 0$.

As we assume that $T$ is diagonal, the singular values of $T$ are simply $\left|T_{11}\right|,\left|T_{22}\right|$, and $\left|T_{33}\right|$. The fully entangled fraction of $\rho_{d}$ is then given by \cite{localEnvironmentFidelity}
\begin{equation}
\label{fidForm}
F(\rho_{d}) = \frac{1}{4} \left(1 + \left|T_{11}\right| +  \left|T_{22}\right| +  \left|T_{33}\right| \right)\ ,
\end{equation}
where we assumed that $\rho_{d}$ is non-local, i.e., $\det(T) < 0$. 

It is easy to see that $NL[\rho] = NL[\rho_{d}]$ because the \emph{local} unitary operations do not change the non-locality of a given state. Additionally, it is also true that $F(\rho) = F(\rho_{d})$ because the fully entangled fraction of $\rho$ is not changed by local unitary operations. Hence, if we prove that $F(\Omega_{\alpha}) \geq F(\rho_{d})$ implies $NL[\Omega_{\alpha}] \geq NL[\rho_{d}]$, the result itself is proved. Furthermore, using Propositions \ref{weakEntangledCHSH} and \ref{weakEntangledFidelity} and equations (\ref{nonlocalityEquation}) and (\ref{fidForm}), it is enough to prove 
\begin{equation}
\label{easierTheorem}
\frac{1+\alpha}{2} \geq \frac{1}{4} \left(1 + \left|T_{11}\right| +  \left|T_{22}\right| +  \left|T_{33}\right| \right) \Rightarrow 2 \sqrt{1 + \alpha^{2}} \geq 2 \sqrt{T_{11}^{2} + T_{22}^{2}} ~.
\end{equation}

In order to be able to prove equation (\ref{easierTheorem}), we first need to derive restrictions on the coefficients of the matrix $T$, because not all matrices $T$ with $T_{ij} \in [-1,1]$ correspond to a valid quantum state, i.e., a density matrix which is a \emph{positive} operator with trace $1$. An operator $\rho_{d}$ is positive if and only if
\begin{equation}
\left\langle \varphi \right| \rho_{d} \left| \varphi \right\rangle \geq 0  ~~ \forall \left|\varphi\right\rangle ~.
\end{equation}
Therefore, we must also have
\begin{equation}
\label{positivPsi}
  \left\langle \Psi^{+} \right| \rho_{d} \left| \Psi^{+} \right\rangle \geq 0 ~.
\end{equation}
Simple calculations show that  (\ref{positivPsi}) implies 
\begin{equation}
\label{T33equation}
T_{33} \leq 1 + T_{11} + T_{22} < 0\ ,
\end{equation}
where the second inequality is an implication of  $T_{11} + T_{22} < -1$ (which follows from $\left|T_{11}\right| + \left|T_{22}\right| > 1$ and $T_{11}, T_{22} \leq 0$).

Solving the left-hand side of equation (\ref{easierTheorem}) for $\alpha$, and using  $T_{11},T_{22},T_{33} \leq 0$, we get
\begin{equation}
\label{posEq}
 \alpha \geq - \frac{1}{2} (1 + T_{11} + T_{22} + T_{33}) ~.
\end{equation}
It is easy to see that $1 + T_{11} + T_{22} + T_{33} < 0$ if we demand $\rho_{d}$ to be non-local and, hence, $\alpha > 0$. Hence, we can conclude that
\begin{equation}
\label{eqnAlpha1}
 1 + \alpha^2 \geq 1 + \frac{1}{4} (1 + T_{11} + T_{22} + T_{33})^2 ~.
\end{equation}
By putting equation (\ref{T33equation}) into  (\ref{eqnAlpha1}), we get
\begin{eqnarray}
\label{eqnAlpha2}
1 + \alpha^2 &\geq& 1 +  (1 + T_{11} + T_{22})^2  \nonumber \\
&=& 2 + 2 T_{11} + 2 T_{22} + 2 T_{11} T_{22} + T_{11}^2 + T_{22}^2 \nonumber \\
&=& 2 \left( (1 + T_{11})(1 + T_{22}) \right) + T_{11}^2 + T_{22}^2 \nonumber \\
&\geq& T_{11}^2 + T_{22}^2 ~.
\end{eqnarray}
This holds because of $1 + T_{11} + T_{22} + T_{33} < 0$ in equation (\ref{eqnAlpha1}) and $1 + T_{11} + T_{22} < 0$ in  (\ref{T33equation}) and (\ref{eqnAlpha2}). The last inequality comes from the fact that $T_{11} \geq -1$ and $T_{22} \geq -1$. Hence, we have proved that $1 + \alpha^2 \geq  T_{11}^2 + T_{22}^2$ and therefore equation (\ref{easierTheorem}).
\end{proof}

Figure \ref{maximalCHSHplot} shows the correspondence between the fully entangled fraction and non-locality of $\Omega_\alpha$ states. Note that pure states have exactly the same relation.
\begin{figure}
 \centering
 \input{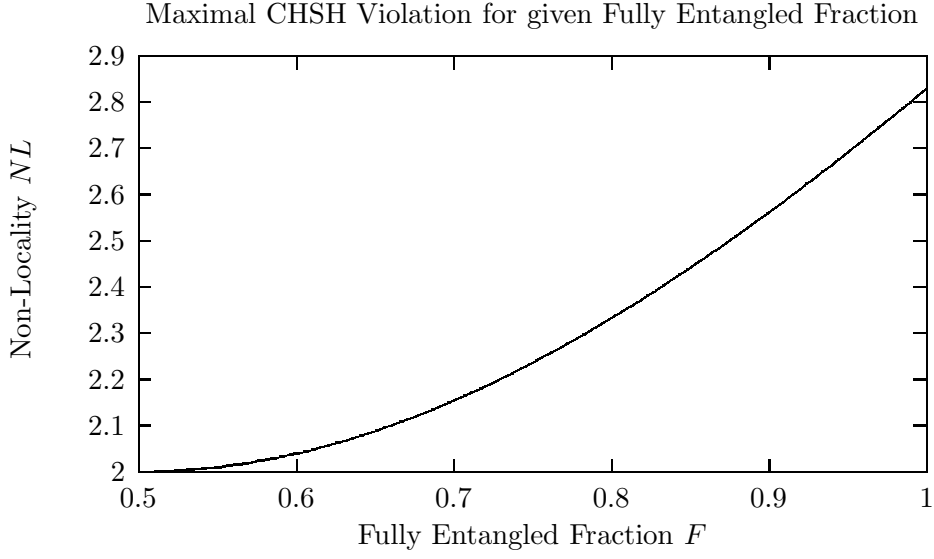}
 \caption{The maximal CHSH violation for given fully entangled fraction is plotted for pure states and weakly entangled mixed states  $\Omega_{\alpha}$ (see Definition \ref{tauDef}). All possible quantum states are located below the curve corresponding to pure states and $\Omega_{\alpha}$ states.}
 \label{maximalCHSHplot}
\end{figure}

\begin{proposition}
\label{propLocalUpper}
Let $\rho$ be an arbitrary two qubit quantum state with $\frac{1}{4} \leq F(\rho) \leq \frac{1}{2}$. Then the non-locality of $\rho$ is upper bounded by
\[
 NL[\rho] \leq 4 F(\rho) ~.
\]
\end{proposition}

\begin{proof}
We will use the same notation as in the proof of Proposition \ref{proMaxNonlocal}. As $F(\rho)$ is smaller or equal than $1/2$, it can be shown \cite{localEnvironmentFidelity} that $\rho$ must be local and therefore the fully entangled fraction can be written as \cite{localEnvironmentFidelity}
\begin{equation}
F(\rho) = \frac{1}{4}\left(1 + \max_{i \neq j \neq k}(|T_{ii}| + |T_{jj}| - |T_{kk}|)\right) ~,~ i,j,k \in \{1,2,3\}~.
\end{equation}
Without loss of generality we can assume that $T_{11}$ and $T_{22}$ have a larger absolute value than $T_{33}$ and we can therefore write for the fully entangled fraction of $\rho$ the simpler expression
\begin{equation}
F(\rho) = \frac{1}{4}\left(1 + |T_{11}| + |T_{22}| - |T_{33}|\right) ~. 
\end{equation}
Furthermore, as $F(\rho) \leq 1/2$ implies \cite{localEnvironmentFidelity} that $|T_{11}| + |T_{22}| + |T_{33}| \leq 1$, we get
\begin{equation}
|T_{11}| + |T_{22}| \leq 2 F(\rho) ~. 
\end{equation}
As we have assumed that $T_{11}$ and $T_{22}$ have a larger absolute value than $T_{33}$, we get for the non-locality of $\rho$ (see proof of Proposition \ref{proMaxNonlocal}) a value of $NL[\rho] = 2 \sqrt{T_{11}^2 + T_{22}^2}$. Hence, we can upper bound the non-locality of $\rho$ by
\begin{equation}
NL[\rho] = 2 \sqrt{T_{11}^2 + T_{22}^2} \leq 2 (|T_{11}| + |T_{22}|) \leq 4 F(\rho) ~,
\end{equation}
which finishes the proof.
\end{proof}

\begin{lemma}
Let $\Omega_{\alpha}$ be a weakly entangled mixed state and $\sigma^{ij}$ be arbitrary two-qubit quantum states and $\sum_{ij} p_{ij} = 1$ with $p_{ij} \geq 0$, then
\[
 F(\Omega_\alpha) \geq \sum_{ij} p_{ij} \cdot F(\sigma^{ij}) \Rightarrow NL[\Omega_{\alpha}] + c(\alpha) \geq \sum_{ij} p_{ij} \cdot NL[\sigma^{ij}] ~,
\]
where $c(\alpha) = \alpha (2 \sqrt{2} - 2) - 2 \sqrt{1 + \alpha^2} + 2$.
\end{lemma}

\begin{proof}
We can rewrite the statement by using Propositions \ref{weakEntangledCHSH} and \ref{weakEntangledFidelity} as follows: if $\frac{1+\alpha}{2} = \sum_{k} p_{k} \cdot F(\sigma^k)$ then $2 \sqrt{1 + \alpha^2} + c(\alpha) \geq \sum_{k} p_k \cdot NL[\sigma^k]$ with $\sum_{k} p_{k} = 1$ and $p_k \geq 0$. Without loss of generality, we assume that for the first $m$ terms in the sum, with $m \in \{0,1,...,n\}$, it holds that $1/4 \leq F(\sigma^k) < 1/2$, for $k \in \{1,2,...,m\}$, and for the remaining terms $1/2 \leq F(\sigma^k) \leq 1$, for $k \in \{m+1,...,n\}$.

Let us make the following substitutions. If $1/4 \leq F(\sigma^k) < 1/2$ then set $x_k \equiv 2 F(\sigma^k) - 1 \in [-1/2,0)$ and if $1/2 \leq F(\sigma^k) \leq 1$ then set $y_k \equiv 2 F(\sigma^k) - 1 \in [0,1]$. In order to prove the lemma we need to find the minimal value $c(\alpha)$ must have such that the statement is true. By using Propositions \ref{proMaxNonlocal} and \ref{propLocalUpper} we can write down the following constrained optimization problem
\begin{eqnarray}
\max_{\{p_k\}, \{x_k\}, \{y_k\}} c(\alpha) &=& - 2\sqrt{1 + \alpha^2} + \sum_{k=1}^{m} p_k \cdot 2 \cdot (1 + x_k) + \sum_{k=m+1}^n p_k \cdot 2 \cdot \sqrt{1 + y_k^2}  \nonumber \\
\textit{subject to:   } \alpha &=& \sum_{k=1}^m p_k \cdot x_k + \sum_{k=m+1}^n p_k \cdot y_k ~~,~~ x_k \in [-1/2,0) ~~,~~ y_k \in [0,1] ~, \nonumber \\
1 &=& \sum_{k=1}^n p_k  ~~,~~ p_k \geq 0 ~.
\end{eqnarray}
One can show that this program is maximized for $p_k = 0$ for all $k \in \{1,2,...,m\}$. Furthermore, only two terms are needed in the second sum, i.e., $p_{m+1} \equiv p$ and $p_{m+2} \equiv 1-p$.
In order to maximize $c(\alpha)$ we have to set $y_{m+1} = 1$ and $y_{m+2} = 0$. Putting these values in the first constraint gives us for $p$ the following value
\begin{equation}
p = \alpha ~.
\end{equation}
Inserting all these values into the objective function of the optimization problem results in
\begin{equation}
c(\alpha) = \alpha (2 \sqrt{2} - 2) - 2 \sqrt{1 + \alpha^2} + 2 ~,
\end{equation}
which proves the lemma.
\end{proof}

\renewcommand{\theequation}{C-\arabic{equation}}
\setcounter{equation}{0}  

\section{Proof of Lemma \ref{simultDiagonal}}

The proof of Lemma \ref{simultDiagonal} is taken from \cite{MasanesBlockDiagonal}.

\begin{lemma}
Let $P_0,P_1,Q_0,Q_1$ be four projectors acting on a Hilbert-space $\mathcal{H}$ such that $P_0 + P_1 = I$ and $Q_0 + Q_1 = I$. There exists an orthonormal basis in $\mathcal{H}$ where the four projectors $P_0,P_1,Q_0,Q_1$ are simultaneously block diagonal, in blocks of size $1 \times 1$ or $2 \times 2$. 
\end{lemma}

\begin{proof}
As the three operators $Q_0$, $(Q_0 P_0 Q_0)$ and $(Q_0 P_1 Q_0)$ do all commute with each other, they can be simultaneously diagonalized \cite{matrixAnalysis}. Let $\vert v \rangle$ be one of their simultaneous eigenvectors. As $Q_0$ and $Q_1$ project onto complementary orthogonal subspaces and $Q_0 \vert v \rangle = \vert v \rangle$, we conclude that $Q_1 \vert v \rangle = 0$. Furthermore, as $P_0 + P_1 = I$, it cannot be the case that $P_0 \vert v \rangle = P_1 \vert v \rangle = 0$.

If $P_0 \vert v \rangle = 0$ then $P_1 \vert v \rangle = \vert v \rangle = Q_0 \vert v \rangle$ and the span of $\vert v \rangle$ (denoted $E_v$) corresponds to a $1 \times 1$ diagonal block in which $P_0, P_1, Q_0, Q_1$ have eigenvalues $0, 1, 1, 0$, respectively. The case $P_1 \vert v \rangle = 0$ is similar.

Consider the case where $P_0 \vert v \rangle \neq 0$ and $P_1 \vert v \rangle \neq 0$. Define the orthogonal vectors $\vert a_1 \rangle = P_0 \vert v \rangle$, $\vert a_2 \rangle = P_1 \vert v \rangle$, and the two-dimensional subspace $E_v = \{\alpha_1 \vert a_1 \rangle + \alpha_2 \vert a_2 \rangle   ~:~ \forall \alpha_1, \alpha_2 \in \mathbb{C} \}$. The fact $\vert v \rangle = \vert a_1 \rangle + \vert a_2 \rangle$ implies $\vert v \rangle \in E_v$. We have 
\begin{eqnarray}
Q_0 \vert a_1 \rangle &=& Q_0 P_0 \vert v \rangle = Q_0 P_0 Q_0 \vert v \rangle = \lambda_1 \vert v \rangle ~, \nonumber \\
Q_0 \vert a_2 \rangle &=& Q_0 P_1 \vert v \rangle = Q_0 P_1 Q_0 \vert v \rangle = \lambda_2 \vert v \rangle \nonumber
\end{eqnarray}
because $\vert v \rangle$ is an eigenvector of $Q_0 P_0 Q_0$ and $Q_0 P_1 Q_0$. The vector $\vert w \rangle = (1 / \lambda_1) \vert a_1 \rangle - (1 / \lambda_2) \vert a_2 \rangle$ is therefore an eigenvector of $Q_0$ and $Q_1$, i.e., $Q_0 \vert w \rangle = 0$ and $Q_1 \vert w \rangle = \vert w \rangle$. Summarizing, the vectors $\vert a_1 \rangle, \vert a_2 \rangle \in E_v$ are simultaneous eigenvectors of $P_0$, $P_1$ and the vectors $\vert v \rangle, \vert w \rangle \in E_v$ are simultaneous eigenvectors of $Q_0$, $Q_1$. Therefore, the subspace $E_v$ corresponds to a $2 \times 2$ simultaneous diagonal block for $P_0, P_1, Q_0, Q_1$. The same can be done with the rest of simultaneous eigenvectors $\vert v \rangle$ as defined above. And analogously, for the simultaneous eigenvectors of $Q_1$, $(Q_1 P_0 Q_1)$ and $(Q_1 P_1 Q_1)$ which are orthogonal to the vectors $\vert w \rangle$ that have appeared in the previous steps. At the end, the direct sum of the subspaces $E_1, E_2, ...$ is $\mathcal{H}$; each subspace $E_i$ of dimension two contains two eigenvectors of each operator $P_0, P_1, Q_0, Q_1$.
\end{proof}


\newpage


\begin{thebibliography}{99}


\bibitem{entanglementDistillationUpperbound}
A.~Ambainis, K.~Yang, Towards the classical communication complexity of entanglement distillation protocols with incomplete information, \emph{ECCC Report} TR03-082 (2003).

\bibitem{localEnvironmentFidelity}
P.~Badzi\c{a}g, M.\ Horodecki, P.\ Horodecki, R.\ Horodecki, 
Local environment can enhance fidelity of quantum teleportation, \emph{Phys.~Rev.~A} \textbf{62}, 012311 (2000). 

\bibitem{kent}
J.\ Barrett, L.\ Hardy, A.\ Kent, No signalling and quantum key distribution, 
{\em Phys.\ Rev.\ Lett.} \textbf{95}, 010503 (2005).

\bibitem{nonLocalCorr}
J.\ Barrett, N.\ Linden, S.\ Massar, S.\ Pironio, S.\ Popescu, D.\ Roberts, \emph{Phys.~Rev.~A} \textbf{71}, 022101 (2005).

\bibitem{bellInequality}
J.~S.~Bell, On the Einstein-Podolsky-Rosen paradox, \emph{Physics} \textbf{1}, 195-200 (1964).

\bibitem{EntanglementDistillationPossible}
C.~H.~Bennett, H.~J.~Bernstein, S.~Popescu, B.~Schumacher, Concentrating partial entanglement by local operations, \emph{Phys.~Rev.~A} \textbf{53}, 2046-2052 (1996).

\bibitem{fidelityCite}
C.~B.~Bennett, D.\ P.\ DiVincenzo, J.\ A.\ Smolin, W.\ K.\ Wootters, Mixed-state entanglement and quantum error correction, \emph{Phys.~Rev.~A} \textbf{54}, 3824 (1996).

\bibitem{world}
G.\ Brassard, H.\ Buhrman, N.\ Linden, A.\ A.\ M\'ethot, A.\ Tapp, F.\ Unger,  A limit on nonlocality in any world in which communication complexity is not trivial,  {\em Physical Review Letters}, Vol.\ 96 (2006).

\bibitem{NLnotQuantum}
A.~Cabello, Proposed experiment to test the bounds of quantum correlations, \emph{Phys.~Rev.~Lett.} \textbf{92}, 060403 (2004).

\bibitem{CHSH}
J.~F.~Clauser, M.~A.~Horne, A.~Shimony, R.~A.~Holt, Proposed experiment to test local hidden-variable theories, \emph{Phys.~Rev.~Lett.} \textbf{23}, 880-884 (1969).



\bibitem{epr}
A.~Einstein, B.~Podolsky, N.~Rosen.
\newblock Can quantum-mechanical description of physical reality be considered
  complete?,
\newblock {\em Physical Review} \textbf{47}, 777--780 (1935).



\bibitem{FosterBox}
M.~Forster, S.~Winkler, S.~Wolf, Distilling non-locality, \emph{arXiv:0809.3173} (2008).


\bibitem{matrixAnalysis}
R.~A.~Horn, C.~R.~Johnson, Matrix analysis, \emph{Cambridge University Press} (1985).

\bibitem{violatingBellInequality}
R.~Horodecki, P.~Horodecki, M.~Horodocki, Violating Bell inequality by mixed spin-1/2 states: necessary and sufficient condition, \emph{Phys.~Lett.~A} \textbf{200}, 340-344 (1995).

\bibitem{MasanesBlockDiagonal}
L.~Masanes, Asymptotic violation of Bell inequalities and distillability, \emph{Phys.~Rev.~Lett.}~\textbf{97}, 050503 (2006).



\bibitem{NonlocalityAxiom}
D.~Rohrlich, S.~Popescu, Nonlocality as an axiom for quantum theory, \emph{Found.~Phys.} \textbf{24}, 379 (1994).



\bibitem{TsirelsonBound}
B.~S.~Tsirelson, Quantum generalization of Bell's inequality, \emph{Lett.~in~Math.~Phys.} \textbf{4}, 93-100 (1980).

\bibitem{DrYang} K.\ Yang, On the (im)possibility of non-interactive correlation distillation, 
{\em Theoretical Computer Science},  Vol.~382, No.~2, pp.~157--166 (2007). 

\end{thebibliography}
\end{document}